\documentclass[10pt,a4paper]{article}
\usepackage[latin1]{inputenc}
\usepackage{amsmath}
\usepackage{amsfonts}
\usepackage{amssymb}
\usepackage{bbm}
\usepackage{tabularx}
\usepackage{graphicx}
\usepackage{color}
\usepackage{amsthm}
\usepackage{mathtools}
\usepackage{lineno}

\graphicspath{{figures/}}

\newtheorem{lemma}{Lemma}
\newtheorem{cor}{Corollary}
\newtheorem{theorem}{Theorem}

\newtheorem{obs}{Observation}

\newtheorem{conjecture}{Conjecture}

\newcommand{\fig}[1]{\figurename~\ref{#1}}

\makeatletter
\g@addto@macro\bfseries{\boldmath}
\makeatother

\usepackage[hidelinks]{hyperref}
\usepackage[blocks]{authblk}
\def\identity{{\sf id}}

\author[1]{Stefan Felsner\thanks{Partially supported by DFG grant FE
    340/11--1.}}
\author[2]{Alexander Pilz\thanks{Supported by a Schr\"odinger
    fellowship of the Austrian Science Fund (FWF): J-3847-N35.}}
\author[3]{Patrick Schnider}
\affil[1]{Institut f\"ur Mathematik, Technische Universit\"at Berlin\\
  \texttt{felsner@math.tu-berlin.de}}
\affil[2]{Institute of SoftwareTechnology, Graz University of Technology, Austria\\
  \texttt{apilz@ist.tugraz.at}}
\affil[3]{Department of Computer Science, ETH Z\"{u}rich, Switzerland\\
  \texttt{patrick.schnider@inf.ethz.ch}}

\advance\hoffset-7mm
\newif\ifap
\apfalse

\title{Arrangements of Approaching Pseudo-Lines}
\begin{document}
\maketitle

\begin{abstract}
We consider arrangements of $n$ pseudo-lines in the Euclidean plane
where each pseudo-line $\ell_i$ is represented by a bi-infinite
connected $x$-monotone curve $f_i(x)$, $x \in \mathbb{R}$, s.t.\ for
any two pseudo-lines $\ell_i$ and $\ell_j$ with $i < j$, the
function $x \mapsto f_j(x) - f_i(x)$ is monotonically decreasing and
surjective (i.e., the pseudo-lines approach each other until they
cross, and then move away from each other).  We show that such
\emph{arrangements of approaching pseudo-lines}, under some aspects,
behave similar to arrangements of lines, while for other aspects,
they share the freedom of general pseudo-line arrangements.
For the former, we prove:
\begin{itemize}
\item There are arrangements of pseudo-lines that are not realizable
  with approaching pseudo-lines.
\item Every arrangement of approaching pseudo-lines has a dual
  generalized configuration of points with an underlying arrangement
  of approaching pseudo-lines.
\end{itemize}
For the latter, we show:
\begin{itemize}
\item There are $2^{\Theta(n^2)}$ isomorphism classes of arrangements
  of approaching pseudo-lines (while there are only
  $2^{\Theta(n \log n)}$ isomorphism classes of line arrangements).
\item It can be decided in polynomial time whether an allowable
  sequence is realizable by an arrangement of approaching
  pseudo-lines.
\end{itemize}
Furthermore, arrangements of approaching pseudo-lines can be
transformed into each other by flipping triangular cells, i.e., they
have a connected flip graph, and every bichromatic arrangement of this
type contains a bichromatic triangular cell.
\end{abstract}

\section{Introduction}
Arrangements of lines and, in general, arrangements of hyperplanes are
paramount data structures in computational geometry whose
combinatorial properties have been extensively studied, partially
motivated by the point-hyperplane duality.  
Pseudo-line arrangements are a combinatorial generalization of line arrangements.
Defined by Levi in 1926 the full potential of working with these structures was first
exploited by Goodman and Pollack.

While pseudo-lines can be considered either as combinatorial or
geometric objects, they also lack certain geometric properties that may
be needed in proofs. The following example motivated
the research presented in this paper.

Consider a finite set of lines that are either red or blue, no two of
them parallel and no three of them passing through the same point.
Every such arrangement has a bichromatic triangle, i.e., an
empty triangular cell bounded by red and blue lines.  This can be
shown using a distance argument similar to Kelly's proof of the
Sylvester-Gallai theorem (see, e.g.,~\cite[p.~73]{proofs_book}). We
sketch another nice proof.
% We sketch, dh wir brauchen keine ablenkenden Details
% Suppose that no two crossings in the
% arrangement have the same $x$-coordinate (otherwise, slightly rotate
% the plane), and that there is a blue crossing above a red line.
Think of the arrangement as a union of two monochromatic arrangements in
colors blue and red. Continuously translate the red arrangement in
positive $y$-direction while keeping the blue arrangement in place.
Eventually the combinatorics of the union arrangement will change with
a triangle flip, i.e., with a crossing passing a line.  The area of
monochromatic triangles is not affected by the motion. Therefore, the
first triangle that flips is a bichromatic triangle in the original
arrangement.  See \figurename~\ref{fig_proof_triangle}~(left).

\begin{figure}
\centering
\includegraphics{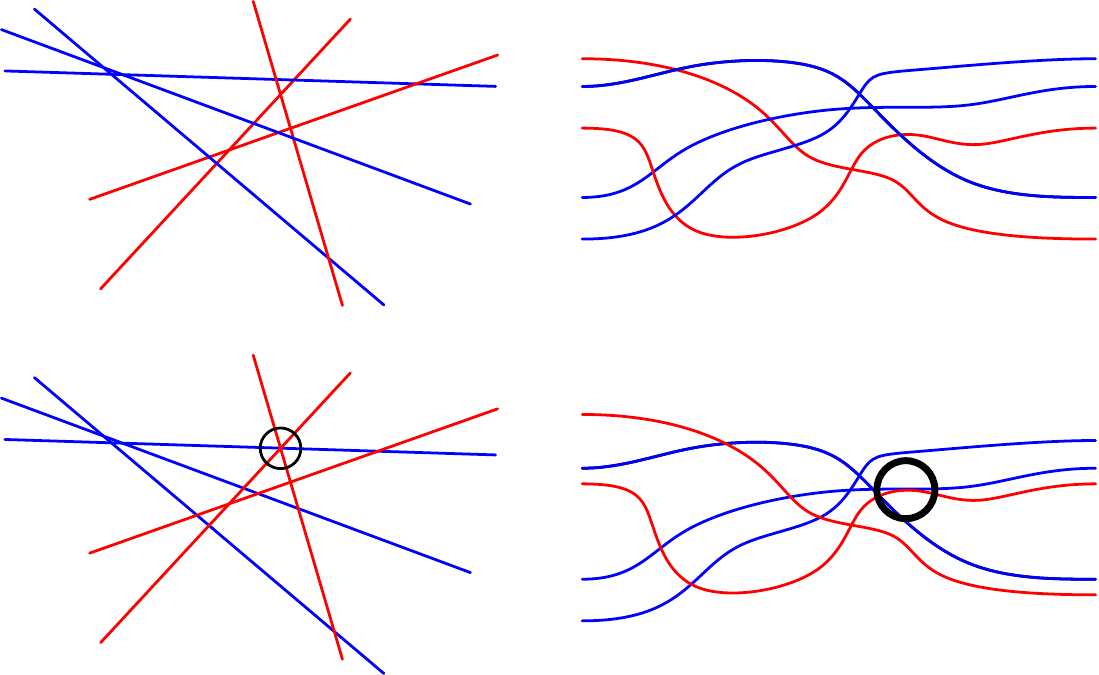}
\caption{Vertical translation of the red lines shows that there is
  always a bichromatic triangle in a bichromatic line arrangement
  (left).  For pseudo-line arrangements, a vertical translation may
  result in a structure that is no longer a valid pseudo-line
  arrangement (right).  }
\label{fig_proof_triangle}
\end{figure}

This argument does not generalize to pseudo-line arrangements.  See
\figurename~\ref{fig_proof_triangle}~(right).  Actually the question
whether all simple bichromatic pseudo-line arrangements have bichromatic
triangles is by now open for several years.  The crucial property of
lines used in the above argument is that shifting a subset of the
lines vertically again yields an arrangement, i.e., the shift does
not introduce multiple crossings.  We were wondering whether any
pseudo-line arrangement can be drawn s.t.\ this property holds.  In
this paper, we show that this is not true and that arrangements where
this is possible constitute an interesting class of pseudo-line
arrangements.

Define an \emph{arrangement of pseudo-lines} as a finite family of
$x$-monotone bi-infinite connected curves (called \emph{pseudo-lines})
in the Euclidean plane s.t.\ each pair of pseudo-lines intersects in
exactly one point, at which they cross.  For simplicity, we consider
the $n$ pseudo-lines $\{\ell_1, \dots, \ell_n\}$ to be indexed from
$1$ to $n$ in top-bottom order at left infinity.%
\footnote{Pseudo-line arrangements are often studied in the real
  projective plane, with pseudo-lines being simple closed curves that
  do not separate the projective plane.  All arrangements can be
  represented by $x$-monotone arrangements~\cite{semispaces}.  As
  $x$-monotonicity is crucial for our setting and the line at infinity
  plays a special role, we use the above definition.}  A pseudo-line
arrangement is \emph{simple} if no three pseudo-lines meet in one
point; if in addition no two pairs of pseudo-lines cross at the same
$x$-coordinate we call it \emph{$x$-simple.}

An \emph{arrangement of approaching pseudo-lines} is an arrangement of
pseudo-lines where each pseudo-line~$\ell_i$ is represented by
function-graph $f_i(x)$, defined for all $x \in \mathbb{R}$, s.t., for
any two pseudo-lines $\ell_i$ and $\ell_j$ with $i < j$, the function
$x \mapsto f_i(x) - f_j(x)$ is monotonically decreasing and
surjective.  This implies that the pseudo-lines approach each other
until they cross, and then they move away from each other, and exactly
captures our objective to vertically translate pseudo-lines in an
arbitrary way while maintaining the invariant that the collection of
curves is a valid pseudo-line arrangement
(If $f_i-f_j$ is not surjective the crossing of pseudo-lines
$i$ and $j$ may be lost upon vertical translations.) 
For most of our results, we consider the pseudo-lines to be
\emph{strictly approaching}, i.e., the function is strictly
decreasing.  For simplicity, we may sloppily call arrangements of
approaching pseudo-lines \emph{approaching arrangements}.

In this paper, we identify various notable properties of approaching
arrangements.  In Section~\ref{sec_manipulating}, we show how to
modify approaching arrangements and how to decide whether an
arrangement is $x$-isomorphic to an approaching arrangement in
polynomial time.  Then, we show a specialization of Levi's enlargement
lemma for approaching pseudo-lines and use it to show that
arrangements of approaching pseudo-lines are dual to generalized
configurations of points with an underlying arrangement of approaching
pseudo-lines.  In Section~\ref{sec_properties}, we describe
arrangements which have no realization as approaching arrangement.  We
also show that asymptotically there are as many approaching
arrangements as pseudo-line arrangements.  We conclude in
Section~\ref{sec_higher} with a generalization of the notion of being
approaching to three dimensions; it turns out that arrangements of
approaching pseudo-planes are characterized by the combinatorial
structure of the family of their normal vectors at all points.

\paragraph{Related work.}
Restricted representations of Euclidean pseudo-line arrangements have
been considered already in early work about pseudo-line arrangements.
Goodman~\cite{goodman_proof} shows that every arrangement has a representation as
a \emph{wiring diagram}. More recently there have been results on drawing arrangements
as convex polygonal chains with few bends~\cite{convex_arc_drawings}
and on small grids~\cite{small_grids}. Goodman and
Pollack~\cite{polynomial_realization} consider arrangements whose
pseudo-lines are the function-graphs of polynomial functions with
bounded degree. In particular, they give bounds on the degree
necessary to represent all isomorphism classes of pseudo-line
arrangements.
Generalizing the setting to higher dimensions (by requiring that any
pseudo-hyperplane can be translated vertically while maintaining that
the family of hyperplanes is an arrangement) we found that such
approaching arrangements are representations of \emph{Euclidean
  oriented matroids}, which are studied in the context of pivot rules
for oriented matroid programming
(see~\cite[Chapter~10]{oriented_matroids}).

\section{Manipulating approaching arrangements}
\label{sec_manipulating}

Lemma~\ref{lem_polygonal} shows that we can make the pseudo-lines of
approaching arrangements piecewise linear. This is similar to the
transformation of Euclidean pseudo-line arrangements to equivalent
wiring diagrams. Before stating the lemma it is appropriate to
briefly discuss notions of isomorphism for arrangements of
pseudo-lines.

%% Diese Frage finde ich gut. Vielleicht bauen wir sie noch irgendwo ein.
%
% Can we always make the first pseudo-line to the last one
% (i.e., change the north pole along the line at infinity)
% and again obtain an approaching arrangement?

Since we have defined pseudo-lines as $x$-monotone curves there are
two faces of the arrangement containing the points at $\pm$infinity of
vertical lines.  These two faces are the \emph{north-face} and the
\emph{south-face}.  A \emph{marked arrangement} is an arrangement
together with a distinguished unbounded face, the
north-face. Pseudo-lines of marked arrangements are oriented such that
the north-face is to the left of the pseudo-line. We think of
pseudo-line arrangements and in particular of approaching arrangements
as being marked arrangements.

Two pseudo-line arrangements are \emph{isomorphic} iff there is an
isomorphism of the induced cell complexes which maps north-face to
north-face and respects the induced orientation of the pseudo-lines.

Two pseudo-line arrangements are \emph{$x$-isomorphic} iff a sweep
with a vertical line meets the crossings in the
same order.

Both notions can be described in terms of allowable sequences. An
\emph{allowable sequence} is a sequence of permutations tarting
with the identity permutation $\identity = (1, \dots, n)$ in which (i) a
permutation is obtained from the previous one by the reversal of one
or more non-overlapping substrings, and (ii) each pair is reversed
exactly once. An allowable sequence is \emph{simple} if two adjacent
permutations differ by the reversal of exactly two adjacent elements.

Note that the permutations in which a vertical sweep line intersects
the pseudo-lines of an arrangement gives an allowable sequence.
We refer to this as \emph{the allowable sequence} of the arrangement and
say that the arrangement \emph{realizes} the allowable sequence.
Clearly two arranements are $x$-isomorphic if they realize the same
allowable sequence.

Replacing the vertical line for the sweep by a moving curve (vertical
pseudo-line) which joins north-face and south-face and intersects each
pseudo-line of the arrangement exactly once we get a notion of
pseudo-sweep. A pseudo-sweep typically has various options for making
progress, i.e., for passing a crossing of the arrangement.  Each
pseudo-sweep also produces an allowable sequence.  Two arrangements
are isomorphic if their pseudo-sweeps yield the same collection of
allowable sequences or equivalently if there are pseudo-sweeps on the
two arrangements which produce the same allowable sequence.

\begin{lemma}\label{lem_polygonal}
  For any arrangement of approaching pseudo-lines, there is an
  $x$-isomorphic arrangement of approaching polygonal curves
  (starting and ending with a ray).  If the allowable sequence of the
  arrangement is simple, then there exists such an arrangement without
  crossings at the bends of the polygonal curves.
\end{lemma}
\begin{proof}
  Consider the approaching pseudo-lines and add a vertical
  `helper-line' at every crossing.  Connect the intersection points of
  each pseudo-line with adjacent helper-lines by segments.  This
  results in an arrangement of polygonal curves between the leftmost
  and the rightmost helper-line.  See \fig{fig_polygonal}.  Since the
  original pseudo-lines were approaching, these curves are approaching
  as well; the signed distance between the intersection points with
  the vertical lines is decreasing, and this property is maintained by
  the linear interpolations between the points.  To complete the
  construction, we add rays in negative $x$-direction starting at the
  intersection points at the first-helper line; the slopes of the rays
  are to be chosen s.t.\ their order reflects the order of the
  original pseudo-lines at left infinity.  After applying the
  analogous construction at the rightmost helper-line, we obtain the
  $x$-isomorphic arrangement.  If the allowable sequence of the
  arrangement is simple, we may choose the helper-lines between the
  crossings and use a corresponding construction.  This avoids an
  incidence of a bend with a crossing.
\end{proof}

\begin{figure}[ht]
\centering
\includegraphics[scale=.8]{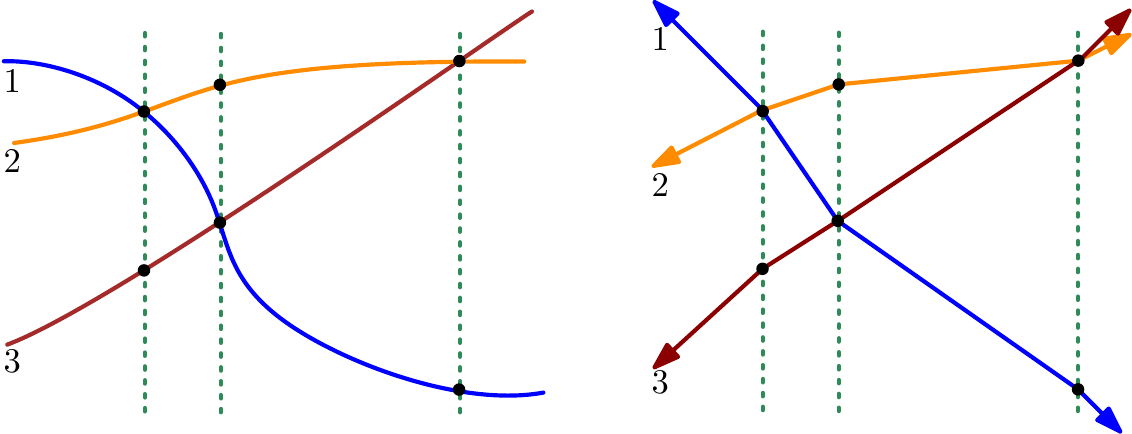}
\caption{Transforming an arrangement of approaching pseudo-lines into
  an isomorphic one of approaching polygonal pseudo-lines.}
\label{fig_polygonal}
\end{figure}

The construction used in the proof yields pseudo-lines being
represented by polygonal curves with a quadratic number of bends.
It might be interesting to consider the problem of minimizing bends in
such polygonal representations of arrangements. 
Two simple operations which can help to reduce the number of bends
are \emph{horizontal streching}, i.e., a change of the $x$-coordinates
of the helper-lines which preserves their left-to-right order,
and \emph{vertical shifts} which can be applied a
helper-line and all the points on it. Both operations
preserve the $x$-isomorphism class.

The two operations are crucial for our next result, where we show
that the intersection points with the helper-lines can be obtained by
a linear program.  Asinowski~\cite{sub_allowable} defines a
\emph{suballowable sequence} as a sequence obtained from an allowable
sequence by removing an arbitrary number of permutations from it.  An
arrangement thus realizes a suballowable sequence if we can obtain
this suballowable sequence from its allowable sequence.

\begin{theorem}\label{thm_realizability}
  Given a suballowable sequence, we can decide in polynomial time
  whether there is an arrangement of approaching pseudo-lines with
  such a sequence.
\end{theorem}
\begin{proof}
  We attempt to construct a polygonal pseudo-line arrangement for the
  given suballowable sequence.  As discussed in the proof of
  Lemma~\ref{lem_polygonal}, we only need to obtain the points in
  which the pseudo-lines intersect vertical helper-lines through
  crossings.  The allowable sequence of the arrangement is exactly the
  description of the relative positions of these points.  We can
  consider the $y$-coordinates of pseudo-line $\ell_i$ at a vertical
  helper-line $v_c$ as a variable $y_{i,c}$ and by this encode the
  suballowable sequence as a set of linear inequalities on those
  variables, e.g., to express that $\ell_i$ is above $\ell_j$ at $v_c$
  we use the inequality $y_{i,c} \geq y_{j,c} +1$.  Further, the
  curves are approaching iff
  $y_{i, c} - y_{j,c} \geq y_{i,c+1} - y_{j, c+1}$ for all
  $1\leq i<j \leq n$ and~$c$.  These constraints yield a polyhedron
  (linear program) that is non-empty (feasible) iff there exists such
  an arrangement.  Since the allowable sequence of an arrangement of
  $n$ pseudolines consists of $\binom{n}{2}+1$ permutations the linear
  program has $O(n^4)$ inequalities in $O(n^3)$ variables.  Note that
  it is actually sufficient to have constraints only for neighboring
  points along the helper lines, this shows that $O(n^3)$
  inequalities are sufficient.
\end{proof}

Let us emphasize that deciding whether an allowable sequence is
realizable by a line arrangement is an $\exists\mathbb{R}$-hard
problem~\cite{allowable_er}, and thus not even known to be in NP.
While we do not have a polynomial-time algorithm for deciding whether
there is an isomorphic approaching arrangement for a given
pseudo-line arrangement, Theorem~\ref{thm_realizability} tells us that
the problem is in NP, as we can give the order of the crossings
encountered by a sweep as a certificate for a realization.  The
corresponding problem for lines is also
$\exists\mathbb{R}$-hard~\cite{mnev}.
% \cite{mnev,shor}.

The following observation is the main property that makes approaching
pseudo-lines interesting.

\begin{obs}\label{obs:main}
  Given an arrangement $A$ of strictly approaching pseudo-lines and a
  pseudo-line $\ell\in A$, any vertical translation of $\ell$ in $A$
  results again in an arrangement of strictly approaching
  pseudo-lines.
\end{obs}

Doing an arbitrary translation, we may run into trouble when the
pseudo-lines are not strictly approaching.  In this case it can happen
that two pseudo-lines share an infinite number of points.  The
following lemma allows us to ignore non-strictly approaching
arrangements for many aspects.

\begin{lemma}\label{lem_strictly}
  Any simple arrangement of approaching pseudo-lines is homeomorphic
  to an $x$-isomorphic arrangement of strictly approaching
  pseudo-lines.
\end{lemma}
\begin{proof}
  Given an arrangement $A$, construct a polygonal arrangement $A'$ as
  described for Lemma~\ref{lem_polygonal}.  If the resulting
  pseudo-lines are strictly approaching, we are done.  Otherwise,
  consider the rays that emanate to the left.  We may change their
  slopes s.t.\ all the slopes are different and their relative order
  remains the same.  Consider the first vertical slab defined by two
  neighboring vertical lines $v$ and $w$ that contains two segments
  that are parallel (if there are none, the arrangement is strictly
  approaching). Choose a vertical line $v'$ slightly to the left of
  the slab and use $v'$ and $w$ as helper-lines to redraw the
  pseudo-lines in the slab.  Since the arrangement is simple the
  resulting arrangement is $x$-isomorphic and it has fewer parallel
  segments. Iterating this process yields the
  desired result.
\end{proof}

\begin{lemma}\label{lem_simple_sweep}
  If $A$ is an approaching arrangement with a non-simple
  allowable sequence, then there exists an approaching arrangement $A'$
  whose allowable sequence is a refinement of the allowable sequence
  of $A$, i.e., the sequence of $A'$ may have additional permutations
  between consecutive pairs $\pi,\pi'$ in the sequence of $A$.
\end{lemma}
\begin{proof}
  Since its allowable sequence is non-simple arrangement $A$ has a
  crossing point where more than two psudo-lines cross or $A$ has several
  crossings with the same $x$-coordinate.  Let $\ell$ be a pseudo-line
  participating in such a degeneracy.  Translating $\ell$ slightly in
  vertical direction a degeneracy is removed and the allowable
  sequence is refined.
\end{proof}

Ringel's homotopy theorem~\cite[Thm.~6.4.1]{oriented_matroids} tells
us that given a pair $A$, $B$ of pseudo-line arrangements, $A$ can be
transformed to $B$ by homeomorphisms of the plane and so-called
\emph{triangle flips}, where a pseudo-line is moved over a crossing.
Within the subset of arrangements of approaching pseudo-lines, the
result still holds.  We first show a specialization of Ringel's
isotopy result~\cite[Prop.~6.4.2]{oriented_matroids}:

\begin{lemma}\label{lem_transform_sweep_equivalent}
  Two $x$-isomorphic arrangements of approaching pseudo-lines can be
  transformed into each other by a homeomorphism of the plane s.t.\
  all intermediate arrangements are $x$-isomorphic and approaching.
\end{lemma}
\begin{proof}
  Given an arrangement~`$A$ of approaching pseudo-lines, we construct
  a corresponding polygonal arrangement $A'$.  Linearly transforming a
  point $f_i(x)$ on a pseudo-line $\ell_i$ in $A$ to the point
  $f'_i(x)$ on the corresponding line $\ell'_i$ in $A'$ gives a
  homeomorphism from $A$ to $A'$ which can be extended to the plane.
  Given two $x$-isomorphic arrangements $A'$ and $B$ of polygonal
  approaching pseudo-lines, we may shift helper-lines horizontally, so
  that the $\binom{n}{2}+1$ helper-lines of the two arrangements
  become adjusted, i.e., are at the same $x$-coordinates; again there
  is a corresponding homeomorphism of the plane.  Now recall that
  these arrangements can be obtained from solutions of linear
  programs.  Since $A'$ and $B$ have the same combinatorial structure,
  their defining inequalities are the same.  Thus, a convex
  combination of the variables defining the two arrangements is also
  in the solution space, which continuously takes us from $A'$ to $B$
  and thus completes the proof.
\end{proof}

\begin{theorem}\label{thm_transform}
  Given two simple arrangements of approaching pseudo-lines, one can
  be transformed to the other by homeomorphisms of the plane and
  triangle flips s.t.\ all intermediate arrangement are approaching.
\end{theorem}
\begin{proof}
  Let $A_0$ be a fixed simple arrangement~$A_0$ of~$n$ lines.  We show
  that any approaching arrangement $A$ can be transformed into $A_0$
  with the given operations.  Since the operations are invertible this
  is enough to prove that we can get from $A$ to $B$.

  Consider a vertical line $v$ in $A$ such that all the crossings of
  $A$ are to the right of $v$ and replace the part of the pseudo-lines
  of $A$ left of $v$ by rays with the slopes of the corresponding
  lines of $A_0$. This replacement is covered by
  Lemma~\ref{lem_transform_sweep_equivalent}. Let $v_0$ be a vertical
  line in $A_0$ which has all the crossings of~$A_0$ to the left. Now
  we vertically shift the pseudo-lines of $A$ to make their
  intersections with $v$ an identical copy of their intersections with
  $v_0$.  During the shifting we have a continuous family of
  approaching arrangements which can be described by homeomorphisms of
  the plane and triangle flips. At the end the order of the
  intersections on $v$ is completely reversed, all the crossings are
  left of $v$ where the pseudo-lines are straight and use the slopes
  of $A_0$. It remains to replace the part of the pseudo-lines
  of $A$ to the right of $v$ by rays with the slopes of the corresponding
\end{proof}

Note that the proof requires the arrangement to be simple.
Vertical translations of pseudo-lines now allows us to prove a
restriction of our motivating question.

\begin{theorem}\label{thm_bichromatic}
  An arrangement of approaching red and blue pseudo-lines contains a
  triangular cell that is bounded by both a red and a blue pseudo-line
  unless it is a pencil, i.e., all the pseudo-lines cross in a single
  point.
\end{theorem}
\begin{proof}
  By symmetry in color and direction we may assume that there is a
  crossing of two blue pseudo-lines above a red pseudo-line.
  Translate all the red pseudo-lines upwards with the same speed.
  Consider the first moment $t>0$ when the isomorphism class
  changes. This happens when a red pseudo-line moves over
  a blue crossing, or a red crossing is moved over a blue pseudo-line.
  In both cases the three pseudo-lines have determined
  a bichromatic triangular cell of the original arrangement.

  Now consider the case that at time $t$ parallel segments of 
  different color are concurrent. In this case we argue as follows.
  Consider the situation at time $\varepsilon>0$ right after the start
  of the motion. Now every multiple crossing is monochromatic
  and we can use an argument as in the proof of
  Lemma~\ref{lem_strictly} to get rid of parallel segments of
  different colors. Continuing the translation after the modification
  reveals a bichromatic triangle as before.
\end{proof}

\section{Levi's lemma for approaching arrangements}

Proofs for showing that well-known properties of line arrangements
generalize to pseudo-line arrangements often use Levi's enlargement
lemma.  (For example, Goodman and Pollack~\cite{helly_type} give
generalizations of Radon's theorem, Helly's theorem, etc.)  Levi's lemma
states that a pseudo-line arrangement can be augmented by a
pseudo-line through any pair of points.  In this section, we show
that we can add a pseudo-line while maintaining the property that all
pseudo-lines of the arrangement are approaching.

\begin{lemma}\label{lem_combination}
  Given an arrangement of approaching pseudo-lines containing two
  pseudo-lines $l_i$ and $l_{i+1}$ (each a function
  $\mathbb{R} \mapsto \mathbb{R}$), consider
  $l' = l'(x) = \lambda l_i(x) + (1-\lambda) l_{i+1}(x)$, for some
  $0 \leq \lambda \leq 1$.  The arrangement augmented by $l'$ is
  still an arrangement of approaching pseudo-lines.
\end{lemma}
\begin{proof}
  Consider any pseudo-line $l_j$ of the arrangement, $j \leq i$.  We
  know that for $x_2 > x_1$,
  $l_i(x_1) - l_j(x_1) \geq l_i(x_2) - l_j(x_2)$, whence
  $\lambda l_j(x_1) - \lambda l_i(x_1) \geq \lambda l_j(x_2) - \lambda l_i(x_2)$.
Similarly, we have $(1-\lambda)l_j(x_1) - (1-\lambda) l_{i+1}(x_1)
\geq (1-\lambda) l_j(x_2) - (1-\lambda) l_{i+1}(x_2)$. 
Adding these two inequalities, we get
\[
l_j(x_1) - l'(x_1) \geq l_j(x_2) -l'(x_2) \enspace .
\]
The analogous holds for any $j \geq i+1$.
\end{proof}

The lemma gives us a means of producing a convex combination of two
approaching pseudo-lines with adjacent slopes.  Note that the adjacency
of the slopes was necessary in the above proof.

\begin{lemma}\label{lem_above}
  Given an arrangement of $n$ approaching pseudo-lines, we can add a
  pseudo-line
  $l_{n+1} = l_{n+1}(x) = l_n(x) + \delta (l_{n}(x) - l_{n-1}(x))$ for
  any $\delta > 0$ and still have an approaching arrangement.
\end{lemma}
\begin{proof}
Assuming $x_2 > x_1$ implies
\[
l_n(x_1) - l_{n+1}(x_1) = l_n(x_1) - l_n(x) - \delta (l_{n}(x_1) - l_{n-1}(x_1)) = \delta(l_{n-1}(x_1) - l_{n}(x_1))
\]\vskip-6mm
\[
\geq \delta(l_{n-1}(x_2) - l_{n}(x_2)) =  l_n(x_2) - l_{n+1}(x_2)\enspace .
\]
With $l_j(x_1) - l_{n}(x_1) \geq l_j(x_2) - l_{n}(x_2)$ we also
get $l_j(x_1) - l_{n+1}(x_1) \geq l_j(x_2) - l_{n+1}(x_2)$ for all
$1\leq j < n$.
\end{proof}

\begin{theorem}\label{thm_approaching_levi}
  Given an arrangement of strictly approaching pseudo-lines and two
  points $p$ and $q$ with different $x$-coordinates, the arrangement
  can be augmented by a pseudo-line $l'$ containing $p$ and $q$ to an
  arrangement of approaching pseudo-lines.  Further, if $p$ and $q$ do
  not have the same vertical distance to a pseudo-line of the initial
  arrangement, then the resulting arrangement is strictly approaching.
\end{theorem}
\begin{proof}
  Let $p$ have smaller $x$-coordinate than $q$.  Vertically translate
  all pseudo-lines such that they pass through $p$ (the pseudo-lines
  remain strictly approaching, forming a pencil through~$p$).  If
  there is a pseudo-line that also passes through $q$, we add a copy
  $l'$ of it. If $q$ is between $l_i$ and $l_{i+1}$, then we find some
  $0<\lambda<1$ such that
  $l'(x) = \lambda l_i(x) + (1-\lambda) l_{i+1}(x)$ contains $p$ and
  $q$. By Lemma~\ref{lem_combination} we can add $l'$ to the
  arrangement.  If $q$ is above or below all pseudo-lines in the
  arrangement, we can use Lemma~\ref{lem_above} to add a pseudo-line;
  we choose $\delta$ large enough such that the new pseudo-line
  contains $q$.  Finally translate all pseudo-lines back to their
  initial position. This yields an approaching extension of the
  original arrangement with a pseudo-line containing $p$ and
  $q$. Observe that the arrangement is strictly approaching unless the
  new pseudo-line was chosen as copy of $l'$.
\end{proof}

Following Goodman et al.~\cite{spread}, a \emph{spread of
  pseudo-lines} in the Euclidean plane is an infinite family of simple
curves such that
\begin{enumerate}
\item each curve is asymptotic to some line at both ends,
\item every two curves intersect at one point, at which they cross, and
\item there is a bijection $L$ from the unit circle $C$ to the family
  of curves such that $L(p)$ is a continuous function (under the
  Hausdorff metric) of $p \in C$.
\end{enumerate}

It is known that every projective arrangement of pseudolines
can be extended to a spread~\cite{spread} (see also
\cite{topological_plane}).
For Euclidean arrangements this is not true because condition 1 
may fail (for an example take the parabolas $(x-i)^2$ as pseudo-lines).
However, given an Euclidean arrangement $A$
we can choose two vertical lines $v_-$ and $v_+$ such that all the
crossings are between $v_-$ and $v_+$ and replace the extensions
beyond the vertical lines by appropriate rays. The reult of
this procedure is called the \emph{truncation} of $A$. Note that
the truncation of $A$ and $A$ are $x$-isomorphic and if $A$ is
approaching then so is the truncation.
We use
Lemma~\ref{lem_combination} to show the following.

\begin{theorem}
  The truncation of every approaching arrangement of pseudo-lines can
  be extended to a spread of pseudo-lines and a single vertical line
  such that the non-vertical pseudo-lines of that spread are
  approaching.
\end{theorem}
\begin{proof}
  Let $l_1, \dots, l_n$ be the pseudo-lines of the truncation of an
  approaching arrangement. Add two almost vertical straight lines
  $l_0$ and $l_{n+1}$ such that the slope of the line connecting two
  points on a pseudoline $l_i$ is between the slopes of $l_0$ and
  $l_{n+1}$.  The arrangement with pseudo-lines
  $l_0,l_1, \dots, l_n,l_{n+1}$ is still approaching. Initialize $S$
  with these $n+2$ pseudolines.  For each $0\leq i \leq n$ and each
  $\lambda \in (0,1)$ add the pseudo-line
  $\lambda l_i(x) + (1-\lambda) l_{i+1}(x)$ to $S$. The proof of
  Lemma~\ref{lem_combination} implies that any two pseudo-lines in $S$
  are approaching. Finally, let~$p$ be the intersection point of
  $l_0$ and $l_{n+1}$ and add all the lines containing $p$ and some
  point above these two lines to $S$. This completes the construction
  of the spread $S$.
\end{proof}

\section{Approaching generalized configurations}

Levi's lemma is the workhorse in the proofs of many properties of
pseudo-line arrangements.  Among these, there is the so-called
\emph{double dualization} by Goodman and Pollack~\cite{semispaces}
that creates, for any arrangement of pseudo-lines, a corresponding
primal generalized configuration of points.

A \emph{generalized configuration of points} is an arrangement of
pseudo-lines with a specified set of $n$ vertices, called
\emph{points}, such that any pseudo-line passes through two points,
and, at each point, $n-1$ pseudo-lines cross.  We assume for
simplicity that there are no other vertices in which more than two
pseudo-lines of the arrangement cross.

Let $\mathcal{C} = (\mathcal{A}, P)$ be a generalized configuration of
points consisting of an approaching arrangement $\mathcal{A}$, and a
set of points $P = \{p_1, \dots, p_n\}$, which are labeled by
increasing $x$-coordinate.  We denote the pseudo-line of $\mathcal{A}$
connecting points $p_i, p_j \in P$ by $p_{ij}$.
%A point triple $(p_i, p_j, p_k)$ is oriented counterclockwise if the point $p_k$ is contained in the pseudo-half-space to the left of the directed pseudo-line $p_i p_j$, and clockwise otherwise.
%(In particular, this matches the usual orientation of a point triple if the arrangement consists of straight lines only.)
%In particular, for $i < j < k$, the set $\{p_i, p_j, p_k\}$ is oriented \emph{positively} if it is oriented counterclockwise.

Consider a point moving from top to bottom at left infinity.  This
point traverses all the pseudo-lines of $\mathcal{A}$ in some order.
We claim that if we start at the top with the identity permutation
$\pi = (1, \dots, n)$, then, when passing $p_{ij}$ we can apply
the (adjacent) transposition $(i,j)$ to $\pi$.
Moreover, by recording all the
permutations generated during the move of the point we obtain an
allowable sequence $\Pi_{\mathcal{C}}$.

Consider the complete graph $K_P$ on the set $P$. Let $c$ be an
unbounded cell of the arrangement $\mathcal{A}$, when choosing $c$ as the
north-face of $\mathcal{A}$ we get a left to right orientation on each
$p_{ij}$. Let this induce the orientation of the edge $\{i,j\}$ of
$K_P$. These orientations constitute a tournament on $P$. It is easy
to verify that this tournament is acyclic, i.e., it induces a
permutation $\pi_c$ on $P$.

\begin{itemize}
\item The order $\pi$ corresponding to the top cell equals
  the left-to-right order on $P$. Since we have labeled the points by
  increasing $x$-coordinate this is the identity.
\item When traversing $p_{ij}$ to get from a cell $c$ to an adjacent
  cell $c'$ the two orientations of the complete graph only differ in
  the orientation of the edge $\{i,j\}$.  Hence, $\pi_c$ and $\pi_c$
  are related by the adjacent transposition $(i,j)$.
\end{itemize}

The allowable sequence $\Pi_{\mathcal{C}}$ and the allowable sequence
of $\mathcal{A}$ are different objects, they differ even in the length
of the permutations.

We say that an arrangement of pseudo-lines is \emph{dual} to a
(\emph{primal}) generalized configuration of points if they have the
same allowable sequence.  Goodman and Pollack~\cite{semispaces} showed
that for every pseudo-line arrangement there is a primal generalized
configuration of points, and vice versa.  We prove the same for the
sub-class of approaching arrangements.

\begin{lemma}\label{lem_dualize}
  For every generalized configuration $\mathcal{C} = (\mathcal{A}, P)$
  of points on an approaching arrangement $\mathcal{A}$, there is an
  approaching arrangement ${A}^*$ with allowable
  sequence~$\Pi_{\mathcal{C}}$.
\end{lemma}
\begin{proof}
  Let $\Pi_{\mathcal{C}} = \pi_0,\pi_1,\ldots,\pi_h$.  We call $(i,j)$
  the adjacent \emph{transposition at $g$} when
  $\pi_g=(i,j)\circ\pi_{g-1}$.  To produce a polygonal approaching
  arrangement~${A}^*$ we define the $y$-coordinates of the
  pseudo-lines $\ell_1,\ldots,\ell_n$ at $x$-coordinates $i\in[h]$.
  Let $(i,j)$ be the transposition at $g$.  Consider the pseudo-line
  $p_{ij}$ of $\mathcal{C}$.  Since $p_{ij}$ is $x$-monotone we can
  evaluate $p_{ij}(x)$.  The $y$-coordinate of the pseudo-line
  $\ell_k$ dual to the point $p_k=(x_k,y_k)$ at $x=g$ is obtained as
  $y_{g}(k) = p_{ij}(x_k)$.

  We argue that the resulting pseudo-line arrangement is approaching.
  Let $(i,j)$ and $(s,t)$ be transpositions at $g$ and $g'$,
  respectively, and assume $g < g'$.  We have to show that
  $y_{g}(a) - y_{g}(b) \geq y_{g'}(a) - y_{g'}(b)$, for all
  $1\leq a < b \leq n$.  From $a < b$ it follows that $p_a$ is left of
  $p_b$, i.e., $x_a < x_b$.  $p_{ij}$ and $p_{st}$ are approaching, we
  get $p_{ij}(x_a) - p_{st}(x_a) \geq p_{ij}(x_b) - p_{st}(x_b)$,
  i.e., $p_{ij}(x_a) - p_{ij}(x_b) \geq p_{st}(x_a) - p_{st}(x_b)$,
  which translates to
  $y_{g}(a) - y_{g}(b) \geq y_{g'}(a) - y_{g'}(b)$.  This completes
  the proof.
\end{proof}

Goodman and Pollack use the so-called \emph{double dualization} to
show how to obtain a primal generalized configuration of points for a
given arrangement ${A}$ of pseudo-lines.  In this process,
they add a pseudo-line through each pair of crossings in
${A}$, using Levi's enlargement lemma.  This results in a
generalized configuration $\mathcal{C}'$ of points, where the points
are the crossings of ${A}$.  From this, they produce the dual
pseudo-line arrangement $\mathcal{A}'$.  Then, they repeat the
previous process for $\mathcal{A}'$ (that is, adding a line through
all pairs of crossings of~$\mathcal{A}'$).  The result is a
generalized configuration $\mathcal{C}$ of points, which they show
being the primal generalized configuration of~$\mathcal{A}$.  With
Theorem~\ref{thm_approaching_levi} and Lemma~\ref{lem_dualize}, we
know that both the augmentation through pairs of crossings and the
dualization process can be done such that we again have approaching
arrangements, yielding the following result.

\begin{lemma}\label{lem_primalize}
  For every arrangement of approaching pseudo-lines, there is a primal
  generalized configuration of points whose arrangement is also
  approaching.
\end{lemma}

Combining Lemmas~\ref{lem_dualize} and~\ref{lem_primalize}, we obtain
the main result of this section.

\begin{theorem}
  An allowable sequence is the allowable sequence of an approaching
  generalized configuration of points if and only if it is the
  allowable sequence of an approaching arrangement.
\end{theorem}

\section{Realizability and counting}\label{sec_properties}

Considering the freedom one has in constructing approaching
arrangements, one may wonder whether actually all pseudo-line
arrangements are $x$-isomorphic to approaching arrangements.  
As we will see in this section, this is not the case.  We use the
following lemma, that can easily be shown using the construction for
Lemma~\ref{lem_polygonal}.

\begin{lemma}\label{lem_three_like_lines}
  Given a simple suballowable sequence of permutations
  $(\identity, \pi_1, \pi_2)$, where $\identity$ is the identity
  permutation, the suballowable sequence is realizable with an
  arrangement of approaching pseudo-lines if and only if it is
  realizable as a line arrangement.
\end{lemma}
\begin{proof}
  Consider any realization~$A$ of the simple suballowable sequence
  with an arrangement of approaching pseudo-lines.  Since the
  arrangement is simple, we can consider the pseudo-lines as being
  strictly approaching, due to Lemma~\ref{lem_strictly}.  There exist
  two vertical lines $v_1$ and $v_2$ s.t.\ the order of intersections
  of the pseudo-lines with them corresponds to $\pi_1$ and $\pi_2$,
  respectively.  We claim that replacing pseudo-line $p_i\in A$ by the
  line $\ell_i$ connecting the points $(v_1,p_i(v_1))$ and
  $(v_2,p_i(v_2))$ we obtain a line arrangement representing the
  suballowable sequence $(\identity, \pi_1, \pi_2)$.

  To prove the claim we verify that for $i < j$ the slope of $\ell_i$
  is less than the slope of $\ell_j$.  Since $A$ is approaching we
  have $p_i(v_1) - p_j(v_1) \geq p_i(v_2) - p_j(v_2)$, i.e.,
  $p_i(v_1) - p_i(v_2) \geq p_j(v_1) - p_j(v_2)$.  The slopes of
  $\ell_i$ and $\ell_j$ are obtained by dividing both sides of this
  inequality by $v_1-v_2$, which is negative.
\end{proof}

Asinowski~\cite{sub_allowable} identified a suballowable sequence
$(\identity, \pi_1, \pi_2)$, with permutations of six elements which is
not realizable with an arrangement of lines.

\begin{cor}\label{cor_no_suballowable}
  There exist simple suballowable sequences that are not realizable by
  arrangements of approaching pseudo-lines.
\end{cor}

With the modification of Asinowski's example shown in
\figurename~\ref{fig_non_realizable}, we obtain an arrangement not
having an isomorphic approaching arrangement.  The modification adds
two almost-vertical lines crossing in the north-cell s.t.\ they form a
wedge crossed by the lines of Asinowski's example in the order of
$\pi_1$.  We do the same for $\pi_2$.  The resulting object is a
simple pseudo-line arrangement, and each isomorphic arrangement
contains Asinowski's sequence.

\begin{figure}
\centering
\includegraphics{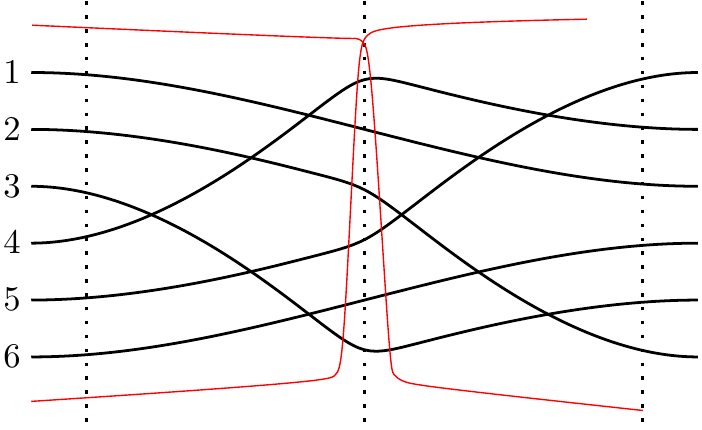}
\caption{A part of a six-element pseudo-line arrangement (bold) whose
  suballowable sequence (indicated by the vertical lines) is
  non-realizable (adapted from~\cite[Fig.~4]{sub_allowable}).  Adding
  the two thin pseudo-lines crossing in the vicinity of the vertical
  line crossed by the pseudo-lines in the order of $\pi_1$ and doing
  the same for $\pi_2$ enforces that the allowable sequence of any
  isomorphic arrangement contains the subsequence
  $(\identity, \pi_1, \pi_2)$.}
\label{fig_non_realizable}
\end{figure}

\begin{cor}
  There are pseudo-line arrangements for which there exists no
  isomorphic arrangement of approaching pseudo-lines.
\end{cor}

Aichholzer et al.~\cite{monotone_paths} construct a suballowable
sequence $(\identity, \pi_1, \pi_2)$ on $n$ lines s.t.\ all line
arrangements realizing them require slope values that are exponential
in the number of lines.  Thus, also vertex coordinates in a polygonal
representation as an approaching arrangement are exponential in $n$.

% Note that, as pointed out in~\cite{monotone_paths}, this is no
% contradiction to them being the result of a linear program.  (and
% thus to the polynomial-time algorithm of
% Theorem~\ref{thm_realizability}), as the coordinates are represented
% in binary.  Cramer + Hadamard inequality

Ringel's Non-Pappus arrangement~\cite{ringel} shows that there are
allowable sequences that are not realizable by straight lines.  It is
not hard to show that the Non-Pappus arrangement has a realization
with approaching pseudo-lines. We will show that in fact the number of
approaching arrangements, is asymptotically larger than the number of
arrangements of lines.

%\begin{conj}
%  For any allowable sequence, realizability by arrangements of
%  approaching pseudo-lines is the same within the projective class.
%\end{conj}

%\section{Combinatorics}\label{sec_combinatorics}

\begin{theorem}\label{thm_number}
  There exist $2^{\Theta(n^2)}$ isomorphism classes of simple
  arrangements of $n$ approaching pseudo-lines.
\end{theorem}
\begin{proof}
  The upper bound follows from the number of non-isomorphic
  arrangements of pseudo-lines.  Our lower-bound construction is an
  adaptation of the construction presented by
  Matou\v{s}ek~\cite[p.~134]{Matousek} for general pseudo-line
  arrangements.  See the left part \fig{fig_lower_bound} for a sketch
  of the construction.  We start with a construction containing
  parallel lines that we will later perturb.  Consider a set $V$ of
  vertical lines $v_i : x = i$, for $i \in [n]$.  Add horizontal
  pseudo-lines $h_i : y = i^2$, for $i \in [n]$.  Finally, add
  parabolic curves $p_i : y = (x + i)^2 - \varepsilon$, defined for
  $x \geq 0$, some $0 < \varepsilon \ll 1$, and $i \in [n]$ (we will
  add the missing part towards left infinity later).  Now, $p_i$
  passes slightly below the crossing of $h_{i+j}$ and $v_j$ at
  $(j,(i+j)^2)$.  See the left part \fig{fig_lower_bound} for a sketch
  of the construction.  We may modify $p_i$ to pass above the crossing
  at $(j,(i+j)^2)$ by replacing a piece of the curve near this point
  by a line segment with slope $2(i+j)$; see the right part of
  \fig{fig_lower_bound}.  Since the derivatives of the parabolas are
  increasing and the derivatives of $p_{i+1}$ at $j - 1$ and of
  $p_{i-1}$ at $j + 1$ are both $2(j+i)$ the vertical distances from
  the modified $p_i$ to $p_{i+1}$ and $p_{i-1}$ remain increasing,
  i.e., the arrangement remains approaching.

  For each crossing $(j,(i+j)^2)$, we may now independently decide
  whether we want $p_i$ to pass above or below the crossing.  The
  resulting arrangement contains parallel and vertical lines, but no
  three points pass through a crossing.  This means that we can
  slightly perturb the horizontal and vertical lines s.t.\ the
  crossings of a horizontal and a vertical remain in the vicinity of
  the original crossings, but no two lines are parallel, and no line
  is vertical.  To finish the construction, we add rays from the
  points on $p_i$ with $x=0$, each having the slope of $p_i$ at $x=0$.
  Each arrangement of the resulting class of arrangements is
  approaching.  We have $\Theta(n^2)$ crossings for which we make
  independent binary decisions. Hence the class consists of
  $2^{\Theta(n^2)}$ approaching arrangements of $3n$ pseudo-lines.
\end{proof}

\begin{figure}
\centering
\includegraphics{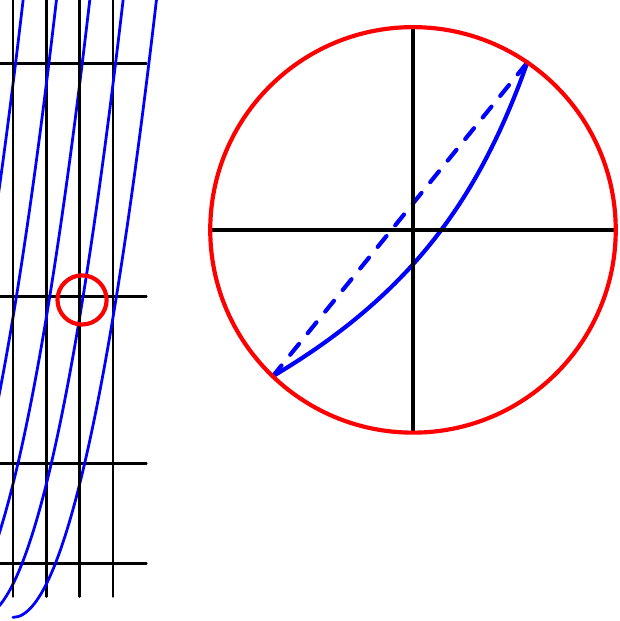}
\caption{A construction for an $2^{\Omega(n^2)}$ lower bound on the
  isomorphism classes of approaching arrangements.}
\label{fig_lower_bound}
\end{figure}

As there are only $2^{\Theta(n \log n)}$ isomorphism classes of simple
line arrangements~\cite{upper_bounds_configurations}, we see that we
have way more arrangements of approaching pseudo-lines.

The number of allowable sequences is
$2^{\Theta(n^2 \log n)}$~\cite{stanley}. We show next that despite of
the existence of nonrealizable suballowable sequences
(Corollary~\ref{cor_no_suballowable}), the number of allowable
sequences for approaching arrangements, i.e., the number of
$x$-isomorphism classes of these arrangements, is 
asymptotically the same as the number of all allowable sequences.

\begin{theorem}
  There are $2^{\Theta(n^2 \log n)}$ allowable sequences realizable as
  arrangements of approaching pseudo-lines.
\end{theorem}
\begin{proof}
  The upper bound follows from the number of allowable sequences.  For
  the lower bound, we use the construction in the proof of
  Theorem~\ref{thm_number}, but omit the vertical lines.  Hence, we
  have the horizontal pseudo-lines $h_i : y = i^2$ and the paraboloid
  curves $p_i : y = (x + i)^2 - \varepsilon$, defined for $x \geq 0$
  and $0 < \varepsilon \ll 1$.  For a parabolic curve $p_i$ and a
  horizontal line $h_{i+j}$, consider the neighborhood of the point
  $(j,(i+j)^2)$.  Given a small value $\alpha$ we
  can replace   a piece of $p_i$ by the appropriate line segment of slope $2(i+j)$
  such that the crosssing of $h_{i+j}$ and the modified $p_i$ 
  has $x$-coordinate $j-\alpha$.

  For fixed $j$ and any permutation $\pi$ of $[n-j]$ we can define
  values $\alpha_i$ for $i \in [n-j]$ such that
  $\alpha_{\pi(1)} < \alpha_{\pi(2)} < \ldots \alpha_{\pi(n-j)}$.
  Choosing the offset values $\alpha_i$ according to different permutations $\pi$
  yields different vertical permutations in the neighborhood of $x=j$, i.e., the
  allowable sequences of the arrangements differ.
  Hence, the number allowable sequences of approaching arrangements
  is at least the superfactorial
  $\prod_{j=1}^{n} j!$, which is in $2^{\Omega(n^2 \log n)}$.
\end{proof}

We have seen that some properties of arrangements of lines are
inherited by approaching arrangements.  It is known that every simple
arrangement of pseudo-lines has $n-2$ triangles, the same is true for
non-simple non-trivial arrangements of lines, however, there are
non-simple non-trivial arrangements of pseudo-lines with fewer
triangles, see~\cite{felsner_kriegel}.  We conjecture that in this
context approaching arrangements behave like line arrangements.

\begin{conjecture}
  Every non-trivial arrangement of $n$ approaching pseudo-lines has at
  least $n-2$ triangles.
\end{conjecture}

\section{Higher dimensions}\label{sec_higher}
An \emph{arrangement of pseudo-hyperplanes} in $\mathbb{R}^d$ is a finite set
$A$ of hypersurfaces, each homeomorphic to $\mathbb{R}^{d-1}$, with the
property that any $k\leq d$ of them intersect as $k$ hyperplanes (no two of
them parallel) do.  More formally, for any $h_1,\ldots,h_k\in A$, $k\leq n$,
the cell complex induced by $h_1,\ldots,h_k$ is isomorphic to
the cell complex of ${\bf e}^\perp_1,\ldots,{\bf e}^\perp_k$
where ${\bf e}^\perp_i$ is the hyperplane whose normal is the $i$th
vector ${\bf e}_i$ of the standard basis.

We focus on arrangements of $\emph{pseudo-planes}$
in $\mathbb{R}^3$. We define arrangements of approaching pseudo-planes
via one of the key properties observed for arrangements of
approaching pseudo-lines (Observation~\ref{obs:main}).

An \emph{arrangement of approaching pseudo-planes} in $\mathbb{R}^3$
is an arrangement of pseudo-planes $h_1,\ldots,h_n$ where each
pseudo-plane $h_i$ is the graph of a continuously differentiable
function $f_i: \mathbb{R}^2\mapsto \mathbb{R}$ such that for any
$c_1,\ldots,c_n\in \mathbb{R}$, the graphs of
$f_1+c_1,\ldots, f_n+c_n$ form a valid arrangement of pseudo-planes.
This means that we can move the pseudo-planes up and down along the
$z$-axis while maintaining the properties of a pseudo-plane
arrangement.  Clearly, arrangements of planes (without parallels) are
approaching.

Consider an arrangement of approaching pseudo-lines with pseudo-lines given by
continuously differentiable functions. The condition that $f_1(x)-f_2(x)$ is strictly
monotonically decreasing implies that for any $x$ the slope of $f_1(x)$ is at
most the slope of $f_2(x)$, where at some single points, they might be equal, e.g. at
$x=0$ for $f_1(x)=0$ and $f_2(x)=x^3$. In other words for each $x$ the
identity permutation is the sorted order of the slopes of the tangents at
$x$. Note that we may think of permutations as labeled Euclidean order types
in one dimension.

In this section we show an analogous characterization of approaching
arrangements of pseudo-planes: The two-dimensional order type associated to
the tangent planes above a point $(x,y)$ is the same except for a sparse set
of exceptional points where the order type may degenerate.

Let $G$ be a collection of graphs of continuously differentiable functions
$f_i: \mathbb{R}^2\mapsto \mathbb{R}$.  For any point $(x,y)$ in
$\mathbb{R}^2$, let $n_i(x,y)$ be the upwards normal vector of the tangent
plane of $f_i$ above $(x,y)$. We consider the vectors $n_i(x,y)$ as points $p_i(x,y)$
in the plane with homogeneous coordinates.  (That is, for each vector we
consider the intersection of its ray with the plane $z=1$.)  We call
$p_i(x,y)$ a \emph{characteristic point} and let $P_G(x,y)$ be the set of
characteristic points. The Euclidean order type of the point
multiset $P_G(x,y)$ is the \emph{characteristic order type} of $G$ at $(x,y)$,
it is denoted $\chi_G(x,y)$.

We denote by $\chi_G$ the set of characteristic order types of $G$ on the
whole plane, that is, $\chi_G=\{\chi_G(x,y)|(x,y)\in\mathbb{R}^2\}$.  We say
that $\chi_G$ is \emph{admissible} if the following conditions hold:

\begin{enumerate}
\item[(1)] for any two points $(x_1,y_1)$ and $(x_2,y_2)$ in the plane, we
  have that if an ordered triple of characteristic points in $P_G(x_1, y_1)$ is
  positively oriented, then the corresponding triple in
  $P_G(x_2,y_2)$ is either positively oriented or collinear;
\item[(2)] for any triple $p_1,p_2,p_3$ of characteristic points, the set of
  points in the plane for which $p_1, p_2, p_3$ are collinear is either the
  whole plane or a discrete set of points (i.e, for each $(x,y)$ in this set
  there is some $\varepsilon >0$ such that the $\varepsilon$-disc around $(x,y)$
  contains no further point of the set);
\item[(3)] for any pair $p_1,p_2$ of characteristic points, the set of
  points in the plane for which $p_1=p_2$ has dimension 0 or 1 (this
  implies that for each $(x,y)$ in this set and each $\varepsilon >0$
  the $\varepsilon$-disc around $(x,y)$ contains points which are not
  in the set).
\end{enumerate}

From the above conditions, we deduce another technical but useful property of
admissible characteristic order types.

\begin{lemma}\label{lem:2b}
  Let $\chi_G$ be an admissible order type and $|G|\geq 3$.
  For any pair $p_1,p_2\in \chi_G$ and
  for every point $(x_0,y_0)$ in the plane for which $p_1=p_2$ there is a 
  neighborhood $N$ such that for $V =\{p_2(x,y)-p_1(x,y) : (x,y)\in N\}$,
  the positive hull of $V$ contains no line.
\end{lemma}
\begin{proof}
  Choose $p_3$ such that $p_3(x_0,y_0) \neq p_1(x_0,y_0) = p_2(x_0,y_0)$.  In
  a small neighborhood $N$ of $(x_0,y_0)$ point $p_3$ will stay away from the
  line spanned by $p_1$ and $p_2$ (continuity). If in $N$ the positive hull of
  $V$ contains a line, then the orientation of $p_1,p_2,p_3$ changes from
  positive to negative in $N$, this contradicts condition~(1) of admissible
  characteristic order types.
\end{proof}

\begin{theorem}
  Let $G$ be a collection of graphs of continuously differentiable functions
  $f_i: \mathbb{R}^2\mapsto \mathbb{R}$.  Then $G$ is an arrangement of
  approaching pseudo-planes if and only if $\chi_G$ is admissible and all the
  differences between two functions are surjective.
\end{theorem}
\begin{proof}
  Note that being surjective is a necessary condition for the difference of
  two functions, as otherwise we can translate them until they do not
  intersect.  Thus, in the following, we will assume that all the differences
  between two functions are surjective.  We first show that if $\chi_G$ is
  admissible then $G$ is an arrangement of approaching pseudo-planes.  Suppose
  $G$ is not an arrangement of approaching pseudo-planes.  Suppose first that
  there are two functions $f_1$ and $f_2$ in $G$ whose graphs do not intersect
  in a single pseudo-line.  Assume without loss of generality that $f_1=0$,
  i.e., $f_1$ is the constant zero function.  Let $f_1\cap f_2$ denote the
  intersection of the graphs of $f_1$ and $f_2$.  If the intersection has a
  two-dimensional component, the normal vectors of the two functions are the
  same for any point in the relative interior of this component, which
  contradicts condition~(3), so from now on, we assume that $f_1\cap f_2$ is
  at most one-dimensional.  Also, note that due to the surjectivity of
  $f_2-f_1$, the intersection $f_1\cap f_2$ is not empty.  Note that if
  $f_1\cap f_2$ is a single pseudo-line then for every $r\in f_1\cap f_2$
  there exists a neighborhood $N$ in $f_1$ such that $f_1\cap f_2\cap N$ is a
  pseudo-segment.
  Further, on one side of the pseudo-segment, $f_1$ is below $f_2$, and above on the other, as otherwise we would get a contradiction to Lemma~\ref{lem:2b}.
  In the next two paragraphs we argue that indeed
  $f_1\cap f_2$ is a single pseudo-line. In paragraph (a) we show that for
  every $r\in f_1\cap f_2$ the intersection locally is a pseudo-segment; in
  (b) we show that $f_1\cap f_2$ contains no cycle and that $f_1\cap f_2$ has
  a single connected component.
  
  (a) Suppose for the sake of contradiction that $f_1\cap f_2$ contains a
  point $r$ such that for every neighborhood $N$ of $r$ in $f_1$ we have that
  $f_1\cap f_2\cap N$ is not a pseudo-segment.  
  For $\varepsilon >0$ let
  $N_\varepsilon$ be the $\varepsilon$-disc around $r$. Consider $\varepsilon$
  small enough such that $f_1\cap f_2\cap N_\varepsilon$ consists of a single
  connected component.
  Further, let $\varepsilon$ be small enough such that whenever we walk away from $r$ in a component where $f_2$ is above (below) $f_1$, the difference $f_2-f_1$ is monotonically increasing (decreasing).
  The existence of such an $\varepsilon$ follows from the fact that $f_1$ and $f_2$ are graphs of continuously differentiable functions.
  Then $f_1\cap f_2$ partitions $N_\varepsilon$ into
  several connected components $C_1,\ldots, C_m$, ordered in clockwise order around $r$. In each of these components, $f_2$ is either
  above or below $f_1$, and this sidedness is different for any two neighboring components.
  In particular, the number of components is even, that is, $m=2k$, for some natural number $k$.
  We will distinguish the cases where $k$ is even and odd, and in both cases we will first show that at $r$ we have $p_1=p_2$ and then apply Lemma~\ref{lem:2b}.
  
  We start with the case where $k$ is even.
  Consider a differentiable path $\gamma$ starting in $C_i$, passing through $r$ and ending in $C_{k+i}$.
  As $k$ is even, $f_2$ is above $f_1$ in $C_i$ if and only if $f_2$ is also above $f_1$ in $C_{k+i}$.
  In particular, the directional derivative of $f_2-f_1$ for $\gamma$ at $r$ is $0$.
  This holds for every choice of $i$ and $\gamma$, thus at $r$ all directional derivatives of $f_2-f_1$ vanish.
  This implies that at $r$ the normal vectors of $f_1$ and $f_2$, coincide, hence $p_1=p_2$.
  Now, consider the boundary of $C_i$.
  Walking along this boundary, $f_2-f_1$ is the constant zero function, and thus the directional derivatives vanish.
  Hence, at any point on this boundary, $p_2-p_1$ must be orthogonal to the boundary, pointing away from $C_i$ if $f_2$ is above $f_1$ in $C_i$, and into $C_i$ otherwise.
  Let now $a$ and $b$ be the intersections of the boundary of $C_i$ with the boundary of $N_\varepsilon$.
  The above argument gives us two directions of vectors, $p_2(a)-p_1(a)$ and $p_2(b)-p_1(b)$, and a set of possible directions of vectors $p_2(c)-p_1(c)$, $c\in C_i$, between them.
  By continuity, all of these directions must be taken somewhere in $C_i$ (see Figure \ref{fig_directions_existence} for an illustration).
  Let now $C_+$ be the set of all components where $f_2$ is above $f_1$, and let $D_+$ be the set of all directions of vectors $p_2(c)-p_1(c)$, $c\in C_+$.
  Further, let $V_+$ be the set of rays emanating from $r$ which are completely contained in $C_+$.
  By continuity, for every small enough $\varepsilon$, there are two rays in $V_+$ which together span a line.
  It now follows from the above arguments, that for these $\varepsilon$, the directions in $D_+$ also positively span a line.
  This is a contradiction to Lemma~\ref{lem:2b}.
  
  \begin{figure}
\centering
\includegraphics{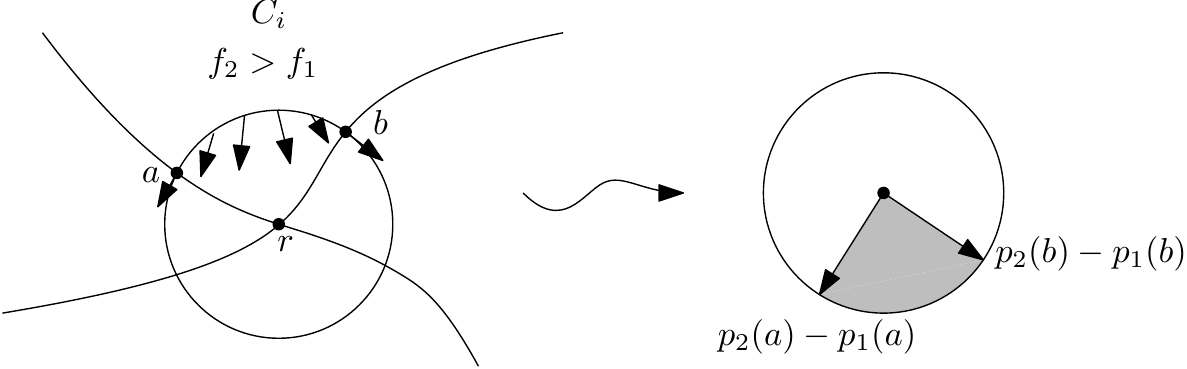}
\caption{A component $C_i$ induces many directions of $p_2-p_1$.}
\label{fig_directions_existence}
\end{figure}
  
  Let us now consider the case where $k$ is odd.
  Consider the boundary between $C_{2k}$ and $C_1$ and denote it by $\gamma_1$.
  Similarly, let $\gamma_2$ be the boundary between $C_k$ and $C_{k+1}$.
  Let now $\gamma$ be the path defined by the union of $\gamma_1$ and $\gamma_2$ and consider the vectors $p_2-p_1$ when walking along $\gamma$.
  Assume without loss of generality that $C_1\in C_+$, and thus $C_{2k}, C_{k+1}\in C_-$ and $C_{k}\in C_+$.
  Analogous to the arguments in the above case, along $\gamma$ the vectors $p_2-p_1$ are orthogonal to $\gamma$, pointing from $C_+$ into $C_-$.
  In particular, they always point to the same side of $\gamma$.
  However, at $r$ the path $\gamma$ is also incident to $C_2\in C_-$ and to $C_{k+2}\in C_+$.
  The same argument now shows that at $r$, the vector $p_2(r)-p_1(r)$ must point from $C_{k+2}$ into $C_2$, that is, into the other side of $\gamma$.
  This is only possible if $p_2(r)-p_1(r)=0$, and thus, as claimed, we again have $p_1=p_2$ at $r$.
  We can now again consider the set of directions $D_+$, and this time, for every small enough $\varepsilon$, the set $D_+$ is the set of all possible directions (see Figure \ref{fig_directions_span} for an illustration), which is again a contradiction to Lemma~\ref{lem:2b}.
  This concludes the proof of claim (a).
  
  \begin{figure}
\centering
\includegraphics{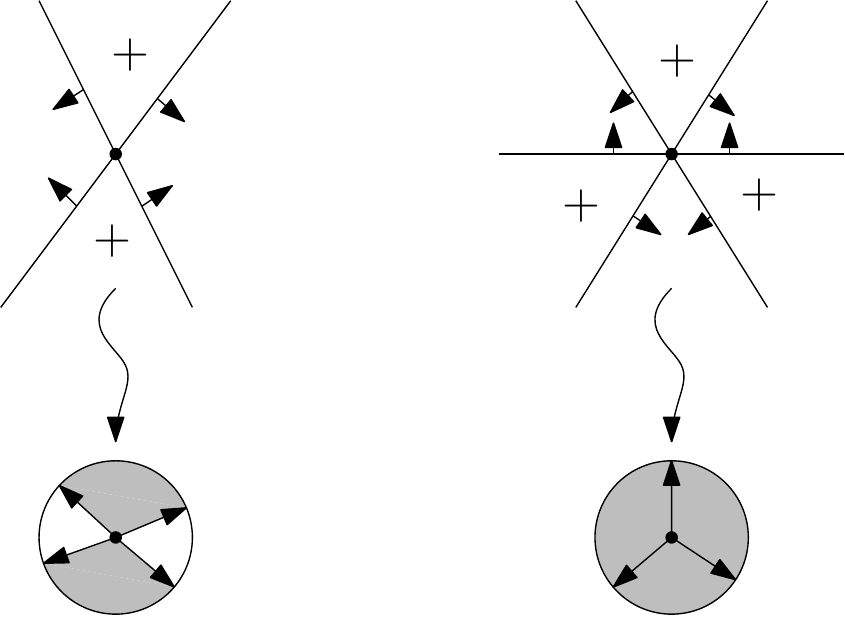}
\caption{$D_+$ spans a line for $k$ even (left) and contains all directions for $k$ odd (right).}
\label{fig_directions_span}
\end{figure}

  (b) Suppose that the intersection $f_1\cap f_2$ contains a cycle.  In the
  interior of the cycle, one function is above the other, so we can vertically
  translate it until the cycle contracts to a point, which again leads to a
  contradiction to Lemma~\ref{lem:2b}.  Now suppose that the intersection
  contains two disjoint pseudo-lines. Between the pseudo-lines, one function
  is above the other, so we can vertically translate it until the pseudo-lines
  cross or coincide.  If they cross, we are again in the case discussed in (a)
  and get a contradiction to Lemma \ref{lem:2b}.  If they coincide, $f_2-f_1$
  has the same sign on both sides of the resulting pseudo-line which again
  leads to a contradiction to Lemma \ref{lem:2b}.
  
  Thus, we have shown that if $\chi_G$ is admissible then any two pseudo-planes
  in $G$ intersect in a single pseudo-line.

  Now consider three functions $f_1,f_2,f_3$ such that any two intersect in a
  pseudo-line but the three do not form a pseudo-hyperplane arrangement.  Then
  in one of the three functions, say $f_1$, the two pseudo-lines defined by
  the intersections with the other two functions do not form an arrangement of
  two pseudo-lines; after translation, we can assume that they touch at a
  point or intersect in an interval.  First assume that they touch at a point.
  At this touching point, one normal vector of tangent planes is the linear
  combination of the other two: assume again without loss of generality that
  $f_1=0$.  Further assume without loss of generality that the curves
  $f_2\cap f_1$ and $f_3\cap f_1$ touch at the point $(0,0)$ and that the
  $x$-axis is tangent to $f_2\cap f_1$ at this point.  Then, as the two curves
  touch, the $x$-axis is also tangent to $f_3\cap f_1$.  In particular, the
  normal vectors to both $f_2$ and $f_3$ lie in the $y$-$z$-plane.  As the
  normal vector to $f_1$ lies on the $z$-axis, the three normal vectors are
  indeed linearly dependent.  For the order type, this now means that one
  vector is the affine combination of the other two, i.e., the three vectors
  are collinear.  Further, on one side of the point the three vectors are
  positively oriented, on the other side they are negatively oriented, which
  is a contradiction to condition~(1).  On the other hand, if they intersect
  in an interval, then the set of points where the vectors are collinear has
  dimension greater than 0 but is not the whole plane, which is a
  contradiction to condition~(2).

  This concludes the proof that if $\chi_G$ is admissible then $G$ is an
  arrangement of approaching pseudo-planes.

  For the other direction consider an approaching arrangement of pseudo-planes
  and assume that $\chi_G$ is not admissible.  First, assume that condition~(1)
  is violated, that is, there are three pseudo-planes $f_1,f_2,f_3$ whose
  characteristic points $p_1,p_2,p_3$ change their orientation from positive
  to negative.  In particular, they are collinear at some point.
  Assume without loss of generality that $f_2$ and $f_3$ are planes containing the origin whose characteristic points are thus constant, and assume without loss of generality that they are $p_2=(0,1)$ and $p_3=(0,-1)$.
  In particular, the intersection of $f_2$ and $f_3$ is the $x$-axis in $\mathbb{R}^3$.
  Consider now a $\varepsilon$-disc $B$ around the origin in $\mathbb{R}^2$ and let $B_<$, $B_0$ and $B_>$ be
  the subsets of $B$ with $x<0$, $x=0$ and $x>0$, respectively.
  Assume without loss of generality that in $B$ the characteristic point $p_1$ is to the left of the $y$-axis
  in $B_<$, to the right in $B_>$, and on the $y$-axis in $B_0$.
  Also, assume that $f_1$ contains the origin in $\mathbb{R}^3$.
  But then, $f_1$ is below the $(x,y)$-plane everywhere in $B$.
  In particular, $f_1$ touches $f_2\cap f_3$ in a single point, namely the origin.
  Hence, $f_1\cap f_3$ and $f_2\cap f_3$ is not an arrangement of two pseudo-lines in $f_3$.
  
 Similar arguments show that
\begin{enumerate}
\item if condition (2) is violated, then after some translation the intersection of some two pseudo-planes in a third one is an interval,
\item if condition (3) is violated, then after some translation the intersection of some two pseudo-planes has a two-dimensional component,
%\item if Lemma~\ref{lem:2b} is violated, then after some translation the intersection of some two pseudo-planes contains a cycle ore several pseudo-lines.\qedhere
\end{enumerate}

\end{proof}

On the other hand, from the above it does not follow to what extent an arrangement of approaching pseudo-planes is determined by its admissible family of characteristic order types.
In particular, we would like to understand which admissible families of order types correspond to families of characteristic order types.
To that end, note that for every graph in an arrangement of approaching pseudo-planes, the characteristic points define a vector
field $F_i: \mathbb{R}^2\mapsto \mathbb{R}^2$, namely its gradient vector
field (a normal vector can be written as $(\text{d}f(x), \text{d}f(y), -1)$.)
In particular, the set of all graphs defines a map $\phi(i,x,y)$ with the
property that $\phi(i,\cdot,\cdot)=F_i$ and the order type of
$\phi(\cdot,x,y)$ is $\chi_G(x,y)$.  We call the family of vector fields
obtained by this map the \emph{characteristic field} of $G$.  A classic result
from vector analysis states that a vector field is a gradient vector field of
a scalar function if and only if it has no curl.  We thus get the following
result:

\begin{cor}\label{characteristic_field}
Let $(F_1,\ldots,F_n)$ be a family of vector fields.
Then $(F_1,\ldots,F_n)$ is the characteristic field of an arrangement of approaching pseudo-planes if and only if each $F_i$ is curl-free and for each $(x,y)\in\mathbb{R}^2$, the set of order types defined by $F_1(x,y),\ldots,F_n(x,y)$ is admissible.
\end{cor}

Let now $G=(g_1,\ldots,g_n)$ be an arrangement of approaching pseudo-planes.
A natural question is, whether $G$ can be extended, that is, whether we can find a pseudo-plane $g_{n+1}$ such that $(g_1,\ldots,g_n,g_{n+1})$ is again an arrangement of approaching pseudo-planes.
Consider the realization of $\chi_G(x,y)$ for some $(x,y)\in\mathbb{R}^2$.
Any two points in this realization define a line.
Let $\mathcal{A}(x,y)$ be the line arrangement defined by all of these lines.
Note that even if $\chi_G(x,y)$ is the same order type for every $(x,y)\in\mathbb{R}^2$, the realization might be different and thus there might be a point $(x',y')\in\mathbb{R}^2$ such that $\mathcal{A}(x',y')$ is not isomorphic to $\mathcal{A}(x,y)$.
For an illustration of this issue, see Figure \ref{fig_admissible_cell}.
(This issue also comes up in the problem of extension of order types, e.g. in \cite{wheels}, where the authors count the number of order types with exactly one point in the interior of the convex hull.)

\begin{figure}
\centering
\includegraphics{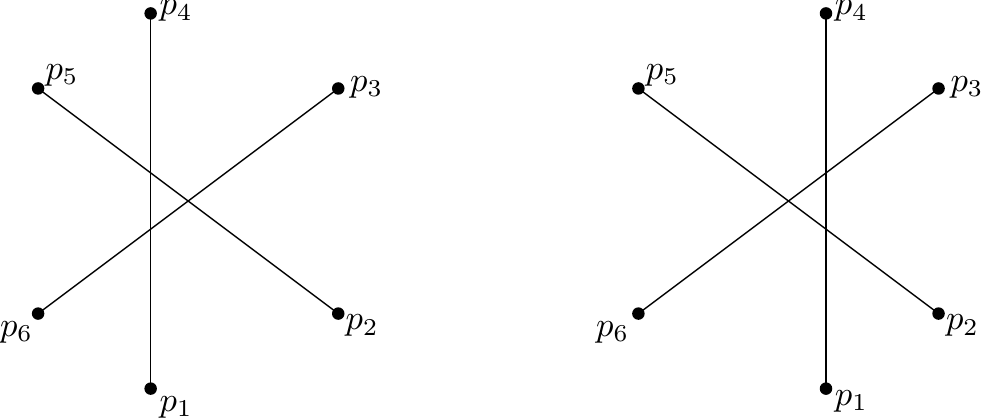}
\caption{Two different arrangements induced by the same order type.}
\label{fig_admissible_cell}
\end{figure}

We call a cell of $\mathcal{A}(x,y)$ \emph{admissible}, if its closure is not empty in $\mathcal{A}(x',y')$ for every $(x',y')\in\mathbb{R}^2$.
Clearly, if we can extend $G$ with a pseudo-plane $g_{n+1}$, then characteristic point $p$ of the normal vector $n_{n+1}(x,y)$ must lie in an admissible cell $c$.
On the other hand, as $c$ is admissible, it is possible to move $p$ continuously in $c$, and if all the vector fields $(F_1,\ldots,F_n)$ are curl-free, then so is the vector field $F_{n+1}$ obtained this way.
Thus, $F_{n+1}$ is the vector field of a differentiable function $f_{n+1}$ and by Corollary \ref{characteristic_field}, its graph $g_{n+1}$ extends $G$.
In particular, $G$ can be extended if and only if $\mathcal{A}(x,y)$ contains an admissible cell.
As the cells incident to a characteristic point are always admissible, we get that every arrangement of approaching pseudo-planes can be extended.
Furthermore, by the properties of approaching pseudo-planes, $g_{n+1}$ can be chosen to go through any given point $p$ in $\mathbb{R}^3$.
In conclusion, we get the following:

\begin{theorem}
Let $G=(g_1,\ldots,g_n)$ be an arrangement of approaching pseudo-planes and let $p$ be a point in $\mathbb{R}^3$.
Then there exists a pseudo-plane $g_{n+1}$ such that $(g_1,\ldots,g_n,g_{n+1})$ is an arrangement of approaching pseudo-planes and $p$ lies on $g_{n+1}$.
\end{theorem}

On the other hand, it could possible that no cell but the ones incident to a characteristic point are admissible, heavily restricting the choices for $g_{n+1}$.
In this case, every pseudo-plane that extends $G$ is essentially a copy of one of the pseudo-planes of $G$.
For some order types, there are cells that are not incident to a characteristic point but still appear in every possible realization, e.g. the unique $5$-gon defined by $5$ points in convex position.
It is an interesting open problem to characterize the cells which appear in every realization of an order type.

\section{Conclusion}
In this paper, we introduced a type of pseudo-line arrangements that generalize line arrangements, but still retain certain geometric properties.
One of the main algorithmic open problems is deciding the realizability of a pseudo-line arrangement as a isomorphic approaching arrangement.
Further, we do not know how projective transformations influence this realizability.
The concept can be generalized to higher dimensions.
Apart from the properties we already mentioned in the introduction, we are not aware of further non-trivial observations.
Eventually, we hope for this concept to shed more light on the differences between pseudo-line arrangements and line arrangements.
For higher dimensions, we gave some insight into the structure of approaching hyperplane arrangements via the order type defined by their normal vectors.
It would be interesting to obtain further properties of this setting.

\bibliographystyle{abbrv}
\bibliography{bibliography}

\begin{thebibliography}{10}

\bibitem{monotone_paths}
O.~Aichholzer, T.~Hackl, S.~Lutteropp, T.~Mchedlidze, A.~Pilz, and
  B.~Vogtenhuber.
\newblock Monotone simultaneous embeddings of upward planar digraphs.
\newblock {\em J. Graph Algorithms Appl.}, 19(1):87--110, 2015.

\bibitem{proofs_book}
M.~Aigner and G.~M. Ziegler.
\newblock {\em Proofs from {THE} {BOOK}}.
\newblock Springer, 5th edition, 2014.

\bibitem{sub_allowable}
A.~Asinowski.
\newblock Suballowable sequences and geometric permutations.
\newblock {\em Discrete Math.}, 308(20):4745--4762, 2008.

\bibitem{oriented_matroids}
A.~Bj\"orner, M.~Las~Vergnas, B.~Sturmfels, N.~White, and G.~Ziegler.
\newblock {\em Oriented Matroids}, volume~46 of {\em Encyclopedia of
  Mathematics and its Applications}.
\newblock Cambridge University Press, 1993.

\bibitem{small_grids}
D.~Eppstein.
\newblock Drawing arrangement graphs in small grids, or how to play planarity.
\newblock {\em J. Graph Algorithms Appl.}, 18(2):211--231, 2014.

\bibitem{convex_arc_drawings}
D.~Eppstein, M.~van Garderen, B.~Speckmann, and T.~Ueckerdt.
\newblock Convex-arc drawings of pseudolines.
\newblock {\em CoRR}, abs/1601.06865, 2016.

\bibitem{felsner_kriegel}
S.~Felsner and K.~Kriegel.
\newblock Triangles in {E}uclidean arrangements.
\newblock {\em Discrete Comput. Geom.}, 22(3):429--438, 1999.

\bibitem{goodman_proof}
J.~E. Goodman.
\newblock Proof of a conjecture of {B}urr, {G}r{\"u}nbaum, and {S}loane.
\newblock {\em Discrete Math.}, 32(1):27--35, 1980.

\bibitem{helly_type}
J.~E. Goodman and R.~Pollack.
\newblock Helly-type theorems for pseudoline arrangements in {$P^2$}.
\newblock {\em J. Comb. Theory, Ser. {A}}, 32(1):1--19, 1982.

\bibitem{semispaces}
J.~E. Goodman and R.~Pollack.
\newblock Semispaces of configurations, cell complexes of arrangements.
\newblock {\em J. Combin. Theory Ser. A}, 37(3):257--293, 1984.

\bibitem{polynomial_realization}
J.~E. Goodman and R.~Pollack.
\newblock Polynomial realization of pseudoline arrangements.
\newblock {\em Commun. Pure Appl. Math.}, 38(6):725--732, 1985.

\bibitem{upper_bounds_configurations}
J.~E. Goodman and R.~Pollack.
\newblock Upper bounds for configurations and polytopes in~{$R^{{d}}$}.
\newblock {\em Discrete Comput. Geom.}, 1:219--227, 1986.

\bibitem{topological_plane}
J.~E. Goodman, R.~Pollack, R.~Wenger, and T.~Zamfirescu.
\newblock Arrangements and topological planes.
\newblock {\em The American Mathematical Monthly}, 101(9):866--878, 1994.

\bibitem{spread}
J.~E. Goodman, R.~Pollack, R.~Wenger, and T.~Zamfirescu.
\newblock Every arrangement extends to a spread.
\newblock {\em Combinatorica}, 14(3):301--306, 1994.

\bibitem{allowable_er}
U.~Hoffmann.
\newblock {\em Intersection graphs and geometric objects in the plane}.
\newblock PhD thesis, Technische Universit{\"a}t Berlin, 2016.

\bibitem{Matousek}
J.~Matou{\v s}ek.
\newblock {\em Lectures on Discrete Geometry}.
\newblock Springer, 2002.

\bibitem{mnev}
N.~E. Mn{\"e}v.
\newblock The universality theorems on the classification problem of
  configuration varieties and convex polytope varieties.
\newblock In O.~Y. Viro, editor, {\em Topology and Geometry---Rohlin Seminar},
  volume 1346 of {\em Lecture Notes Math.}, pages 527--544. Springer, 1988.

\bibitem{wheels}
A.~Pilz, E.~Welzl, and M.~Wettstein.
\newblock From crossing-free graphs on wheel sets to embracing simplices and
  polytopes with few vertices.
\newblock In {\em 33rd International Symposium on Computational Geometry (SoCG
  2017)}. Schloss Dagstuhl-Leibniz-Zentrum fuer Informatik, 2017.

\bibitem{ringel}
G.~Ringel.
\newblock {Teilungen der Ebene durch Geraden oder topologische Geraden}.
\newblock {\em Math.~Z.}, 64:79--102, 1956.

\bibitem{stanley}
R.~P. Stanley.
\newblock On the number of reduced decompositions of elements of {C}oxeter
  groups.
\newblock {\em European J. Combin.}, 5:359--372, 1984.

\end{thebibliography}

\end{document}